\documentclass[a4paper,11pt]{article}
\usepackage[T1]{fontenc}
\usepackage{amsmath, amsfonts, amsthm}
\usepackage{mathrsfs}
\usepackage{xcolor}
\usepackage{bbm}
\usepackage{graphicx}
\usepackage[a4paper]{geometry}
\usepackage{dsfont}
\usepackage{hyperref}
\hypersetup{
colorlinks=true,
linkcolor=blue,
filecolor=magenta,      
urlcolor=cyan
}

\usepackage{array}
\usepackage{tikz}
\usepackage{blindtext}
\usepackage{wrapfig}
\let\originalleft\left
\let\originalright\right
\renewcommand{\left}{\mathopen{}\mathclose\bgroup\originalleft}
\renewcommand{\right}{\aftergroup\egroup\originalright}

\theoremstyle{plain}
\newtheorem{theorem}{Theorem}

\newtheorem{lemma}{Lemma}

\newtheorem{prop}{Proposition}

\theoremstyle{definition}
\newtheorem{defn}{Definition}

\theoremstyle{remark}
\newtheorem{remark}{Remark}
\newtheorem*{remark*}{Remark}
\newtheorem{example}{Example}

\newcommand{\E}{\mathbb{E}}

\newcommand{\R}{\mathbb{R}}

\newcommand{\N}{\mathbb{N}}

\renewcommand{\vec}[1]{\boldsymbol{#1}}

\newcommand{\ra}{\rightarrow}

\newcommand{\bround}[1]{\Big( #1 \Big)}
\newcommand{\sq}[1]{\left[ #1 \right]}
\newcommand{\bsq}[1]{\Big[ #1 \Big]}
\newcommand{\br}[1]{\left\{ #1 \right\}}

\newcommand{\innero}[1]{\langle #1 \rangle}

\newcommand{\ddd}{,\dots,}

\newcommand{\nn}{\nonumber \\}

\newcommand\mdots{\makebox[1em][c]{.\hfil.\hfil.}}

\begin{document}

\title{\textbf{Beyond the I-MMSE relation: Derivatives of mutual information in Gaussian channels}}
\author{Minh-Toan Nguyen \\ GIPSA-lab, Grenoble Alpes University }
\date{}

\maketitle

\begin{abstract}
The I-MMSE formula connects two important quantities in information theory and estimation theory: the mutual information and the minimum mean-squared error (MMSE). It states that in a scalar Gaussian channel, the derivative of the mutual information with respect to the signal-to-noise ratio (SNR) is one-half of the MMSE. Although any derivative at a fixed order can be computed in principle, a general formula for all the derivatives is still unknown. In this paper, we derive this general formula for vector Gaussian channels. The obtained result is remarkably similar to the classic cumulant-moment relation in statistical theory.
\end{abstract}

\section{Introduction}
Consider the following Gaussian channel
\begin{align}\label{sgc}
    Y = \sqrt{\lambda} X + Z,
\end{align}
with output $Y$, input (signal) $ X$ and a standard Gaussian noise $Z$ that is independent of $X$. The non-negative parameter $\lambda$ is called the signal-to-noise ratio (SNR). Let $I_X(\lambda) = I(X;Y)$ be the mutual information between the input $X$ and the output $Y$.  This quantity is linked to the minimum mean-squared error (MMSE) for estimating $X$ from $Y$ by the fundamental I-MMSE formula \cite{guo2005mutual}
\begin{align}
    I'_X(\lambda) = \frac{1}{2} \text{MMSE}(\lambda),
\end{align} 
where
\begin{align}
    \text{MMSE}(\lambda) &= \min_{f \text{ measurable} } \E[(X-f(Y))^2]  \\
    &= \E[(X-\E[X|Y])^2].
\end{align}
Higher derivatives of $I_X(\lambda)$ are also given in \cite{guo2011estimation}. The key idea behind these formulas is the incremental channel approach which reduces the calculation of $I_X^{(k)}(\lambda)$ for any positive $\lambda$ to that of $I_X^{(k)}(0)$. With methods proposed in \cite{guo2005mutual} and \cite{guo2011estimation}, the $k$-th derivatives at zero can be calculated for any $k$, resulting in increasingly complex formulas as $k$ grows. Moreover, \cite{ledoux2016heat} derived a recursive way to compute the derivatives. However, a general formula for all $k$ is currently unknown.

A slightly more general problem is computing the expansion in $\lambda$ of $H(Z_{\lambda})-H(Z)$, called the \emph{neg-entropy} of $Z_\lambda$, where $H(\cdot)$ denotes the entropy and $Z_{\lambda}$ converges in law to the standard normal random variable $Z$ when $\lambda \ra 0$. Some examples of $Z_{\lambda}$ are
\begin{align}
 \sqrt{\lambda} X + Z,\\
\sqrt \lambda X + \sqrt{1-\lambda}Z, \\
 (X_1+\dots+X_n)/\sqrt{n},
\end{align}
where in the last example, $X$ is a random variable with zero mean and unit variance, $X_1 \ddd X_n$ are i.i.d as $X$ and $\lambda = n^{-1/2}$. Of these three examples, the first is equivalent to computing higher derivatives $I_X^{(k)}(0)$, the second is for analyzing the leakage of a protected message in \cite{rioul2021cumulant}, and the third is for approximating the neg-entropy in Independent Component Analysis \cite{comon1994independent} \cite{hyvarinen2000independent}. 

Results for the scalar Gaussian channels generalize to linear vector Gaussian channels. Within this general model and with respect to arbitrary parameters of the model, \cite{palomar2005gradient} and \cite{payaro2009hessian} obtained closed-form expressions for the gradient and the Hessian of the mutual information. Derivatives of the mutual information up to second order are important in proving MMSE crossing properties of parallel Gaussian channels \cite{bustin2012mmse}.

In this work, we derive a general formula for all higher derivatives of the mutual information with respect to the SNRs in a vector Gaussian channel. The obtained formula bears a striking similarity to the classical cumulant-moment relation \cite{speed1983cumulants}, which can also be derived quickly using methods proposed in this paper.

Our work relies on two key components: the reduction of multiple Gaussian channels with identical signals into a single channel without information loss, and the non-rigorous replica method originated from the physics of disordered systems \cite{mezard1987spin}. The replica method also gives a quick derivation of the cumulant-moment formula (Section \ref{cu}).

Several results in this paper, obtained through the non-rigorous replica method, still require rigorous proofs with explicit assumptions. We also make frequent exchange of integrals, expectations and derivatives in the calculations, which are only valid under certain regularity conditions. Currently, we assume that every function and random variable introduced is sufficiently regular, so that operations performed on them are meaningful.

The paper is structured as follows. Section \ref{res} contains the statement of the main result and some of its consequences. In Section \ref{tools}, we present the main tools that are used in our work. The proof of the main result and of other claims made in the paper are given in Section \ref{pf}. 

\textbf{Notation.} 
For any $n \in \N$, the set $\br{x \in \N, 1 \leq x \leq n}$ is denoted as $[n]$ and the multiset $\br{1,1,\ddd,n,n}$, where each elements of $[n]$ is repeated twice, is denoted as $[n]_2$. For vectors $\vec x = (x_1 \ddd x_n)$ and $\vec y = (y_1 \ddd y_n)$ of the same dimension, $\vec x\odot \vec y \equiv (x_1y_1 \ddd x_ny_n)$.

\section{Statement of results}\label{res}
Our results will be stated in terms of multisets and partitions. A \emph{multiset} is a collection of elements where repetitions are allowed, and the arrangement of elements is disregarded. A \emph{partition} of a multiset is a way of dividing it into non-empty multisets, or \emph{blocks}, without taking into account the order of the blocks. If a partition $\pi$ consists of blocks $B_1 \ddd B_k$, we write $\pi = (B_1 \ddd B_k)$ and denote $|\pi| = k$, the number of blocks of $\pi$. A partition is \emph{diverse} if the elements in each of its block are distinct. For example, the partition $ (\br{1,2},\br{1,2})$ of the multiset $ \br{1,1,2,2}$ is diverse while the partition $ (\br{1,1},\br{2,2})$ is not.

For a diverse partition $\pi$ of the multiset $[n]_2$, $s(\pi)$ is defined as the number of \emph{twins} of $\pi$. Here, a twin of $\pi$ is a pair of identical blocks in $\pi$.  For example, $s(\pi) = 2$ for $\pi = (\br{1,2}, \br{1,2}, \br{3,4}, \br{3,4})$ and $s(\pi)=0$ for $\pi= (\br{1,2,3,4}, \br{1,2}, \br{3,4})$. Note that each block in a diverse partition of $[n]_2$ appears at most twice.

With the basic notions of multisets and partitions, let us define the following important object that will allow us to express higher derivatives of the mutual information in a compact way.

\begin{defn}
The form $\tau$, defined with any $n\geq 1$ arguments, is given by:
\begin{align}\label{kk}
&\tau(X_1 \ddd X_n) \nn
\equiv &\sum_{\pi} \frac{(-1)^{|\pi|-1} (|\pi|-2)!}{2^{s(\pi)}}  \, \E[X_{B_1}] \dots \E[X_{B_{|\pi|}}],
\end{align}
where $X_1 \ddd X_n$ are random variables, the sum is taken over all diverse partitions $\pi = (B_1 \ddd B_{|\pi|})$ of $[n]_2$ and $X_B = \prod_{i \in B} X_i$.
\end{defn}

\begin{defn}
The form $\bar \tau$ with $n$ arguments are defined by the same combinatorial sum in (\ref{kk}), except that the sum is over all diverse partitions $\pi$ of $[n]_2$ with all block sizes greater than one.
\end{defn}

\begin{defn}
$ \tau(\,\cdot\,|Y)$ and $\bar \tau(\,\cdot\,|Y)$ where $ Y$ is a random variable or an event, are defined by replacing the expectations $ \E[\,\cdot\,]$ in the definitions of $\tau$ and $\bar \tau$ by $ \E[\,\cdot\,|Y]$.
\end{defn}

\begin{example}\label{lkj}
The multiset $\br{1,1,2,2}$ has the following diverse partitions
\[
(\br{1,2}, \br{1,2}),\quad (\br{1,2}, \br{1}, \br{2})
\]
\begin{align}\label{part}
(\br{1}, \br{1}, \br{2}, \br{2}).
\end{align}
For the partition $(\br{1,2}, \br{1,2})$, we have $|\pi| = 2$ and $s(\pi)=1$, so the coefficient in (\ref{kk}) that corresponds to this partition is $(-1)^{|\pi|-1} (|\pi|-2)!2^{-s(\pi)} = -1/2$.
For the next two partitions, $(|\pi|, s(\pi))$ is given by $(3,0)$ and $(4,2)$ respectively. Therefore
\begin{align}
\tau(X_1, X_2)
= &-\frac{1}{2}\E[X_1 X_2]^2 + \E[X_1 X_2] \E[X_1] \E[X_2] \nn
&- \frac{1}{2} \E[X_1]^2 \E[X_2]^2,
\end{align}
which can be written compactly in term of the covariance as $-\frac{1}{2} \text{Cov}(X_1,X_2)^2$. On the other hand, since only the first partition in (\ref{part}) has all the block sizes greater than one, we have
\begin{align}\label{tau2}
    \bar \tau(X_1, X_2) = -\frac{1}{2}\E[X_1 X_2]^2.
\end{align}
\end{example}

\vspace{.5cm}

The form $\tau$ is remarkably similar to joint cumulants. Recall that the joint cumulant of random variables $X_1 \ddd X_n$ is defined as
\begin{align}
    \kappa(X_1 \ddd X_n) = \partial_{\lambda_1} \dots \partial_{\lambda_n} \psi_{\vec X}(\vec \lambda)|_{\vec \lambda = \vec 0},
\end{align}
where
\begin{align}
    \psi_{\vec X}(\vec \lambda) = \log \E e^{\innero{\vec \lambda, \vec X}}.
\end{align} 
It is clear from the definition that the joint cumulant does not depends on the order of its arguments. The cumulant can be computed from the moments by the classical cumulant-moment formula \cite{speed1983cumulants}:
\begin{align}\label{mc}
&\kappa(X_1 \ddd X_n) \nn 
= &\sum_{\pi} (-1)^{|\pi|-1}(|\pi|-1)! \, \E[X_{B_1}] \dots \E[X_{B_{|\pi|}}],
\end{align}
where the sum is over all partitions $ \pi = (B_1\ddd B_{|\pi|})$ of $[n]$. The joint cumulant is multilinear and $ \kappa((X_i)_{i \in [n]}) = 0$ if $[n]$ can be divided into two non-empty sets $I$ and $J$ such that $ (X_i)_{i \in I}$ and $ (X_j)_{j \in J}$ are independent.

\begin{prop} \label{lt}
$ $

a) $ \tau$ and $\bar \tau$ are multiquadratic.

b) $ \tau((X_i)_{i \in [n]}) = 0$ if $[n]$ can be divided into two disjoint, non-empty sets $I$ and $J$ such that $ (X_i)_{i \in I}$ and $ (X_j)_{j \in J}$ are independent.

c) The forms $\tau$ and $\bar \tau$ are related by the following identity
\begin{align}\label{ttau}
\tau(X_1 \ddd X_n) = \bar \tau(X_1 - \E[X_1] \ddd X_n - \E[X_n]).
\end{align}
\end{prop}
\noindent Here, a \emph{quadratic form} on a vector space $V$ is a function $f: V \ra \R$ such that there exists a bilinear form $ \phi$ such that $f(v) = \phi(v,v)$ for all $v \in V$. A \emph{multiquadratic form} is quadratic in each of its arguments.

While part a of Proposition \ref{lt} is rather obvious, part b and c are highly non-trivial, as proving them directly from the definitions of $\tau$ and $\bar \tau$ will require delving into intricate combinatorial details. However, they can be proved easily using the connection between $\tau$ and mutual information, as we will see in Section \ref{pf}.

The main finding of this study can be stated as follows:
\begin{theorem}\label{main}
Let $\vec X = (X_1 \ddd X_n)$ be a random vector in $\R^n$. Consider the vector Gaussian channel
\begin{align}
\vec Y = \sqrt{\vec \lambda}\odot \vec X + \vec Z,
\end{align}
where $\vec \lambda \in \R_+^n$, $\vec Z$ is independent of $ \vec X$ and follows the standard normal distribution in $\R^n$. Let $I_{\vec X}(\vec \lambda) = I(\vec X; \vec Y)$. Then, for non-negative integers $k_1 \ddd k_n$ such that $k_1+\dots+k_n \geq 2$, we have
\begin{align} 
&\partial_{\lambda_1}^{k_1} \mdots \partial_{\lambda_n}^{k_n} I_{\vec X}(\vec \lambda) \nn 
=&\E[\tau(\underbrace{X_1, \mdots, X_1}_{k_1}, \mdots, \underbrace{X_n, \mdots, X_n}_{k_n} | \vec Y)] \label{jj} \\
=& \E[\bar \tau(\underbrace{\bar X_1, \mdots, \bar X_1}_{k_1}, \mdots, \underbrace{\bar X_n, \mdots, \bar X_n}_{k_n} | \vec Y)] \label{jj1},
\end{align}
where $ \bar X_i = X_i - \E[X_i| \vec Y]$.
\end{theorem}

\begin{remark}
The case $k_1+\dots+k_n =1$ in Theorem \ref{main} is covered by the following result
\begin{align}\label{og}
    \partial_{\lambda_i} I_{\vec X}(\vec \lambda) = \E[(X_i - \E[X_i|\vec Y])^2],
\end{align}
which is a direct consequence of equation (4) in \cite{palomar2005gradient}.
\end{remark}

\begin{remark}
The forms $\tau$ and $\bar \tau$ with $k$ arguments allow us to write down the formulas for all $k$-th derivatives of $I_{\vec X}(\vec \lambda)$, i.e. all the derivatives of the form $\partial_{\lambda_1}^{k_1} \mdots \partial_{\lambda_n}^{k_n} I_{\vec X}(\vec \lambda)$ with $k_1+\dots+k_n=k$.
\end{remark}

Next, we will give some examples for Theorem \ref{main} and recover some results from \cite{guo2011estimation}.

\begin{example}\label{ex:der2}
With the setting and the notations of Theorem \ref{main}, from (\ref{tau2}) and (\ref{jj1}), we have
\begin{align}
    \partial_{\lambda_i} \partial_{\lambda_j} I_{\vec X}(\vec \lambda) = -\frac{1}{2}\E[\, \E[\, \bar X_i \bar X_j | \vec Y \,]^2\, ]
\end{align}
for all $i,j \in [n]$.
\end{example}

\begin{example}\label{ex:der3}
The multiset $\br{1,1,2,2,3,3}$ has the following partitions with all block sizes greater than one
\begin{align}
(\br{1,2,3},\br{1,2,3}), \quad (\br{1,2}, \br{2,3}, \br{3,1}),
\end{align}
which yields the following formula
\begin{align}
    &\bar \tau(X_1, X_2, X_3) \nn
    =&\E[X_1 X_2]\E[X_2 X_3] \E[X_3 X_1] - \frac{1}{2} \E[X_1 X_2 X_3]^2.
\end{align}
From this and (\ref{jj1}), with the setting and the notations of Theorem \ref{main}, we have
\begin{align}
\partial_{\lambda_i} \partial_{\lambda_j} \partial_{\lambda_k} I_{\vec X}(\vec \lambda) &= \E[\bar X_i \bar X_j|\vec Y]\,\E[\bar X_j \bar X_k|\vec Y]\, \E[\bar X_k \bar X_l|\vec Y] \nn
&- \frac{1}{2} \E[\bar X_i \bar X_j \bar X_k|\vec Y]^2,
\end{align}
for any $i,j,k \in [n]$. 
\end{example}

\begin{example} \label{ex:der4}
We have
\begin{align}
	&\bar \tau(X_1, X_2, X_3, X_4) = \nn
	&-2(\E[X_1X_2]\E[X_2 X_3] \E[X_3 X_4]\E[X_4 X_1]+ \text{two other terms}) \nn
	&-\frac{1}{2} (\E[X_1 X_2]^2 \E[X_3 X_4]^2 + \text{two other terms}) \nn
	&+\E[X_1 X_2]\E[X_1 X_3 X_4] \E[X_2 X_3 X_4] + \text{five other terms} \nn
	&+\E[X_1 X_2 X_3 X_4]\E[X_1 X_2]\E[X_3 X_4] + \text{two other terms}\nn
	&-\frac{1}{2}\E[X_1 X_2 X_3 X_4]^2,
\end{align}
from which we can easily write the general formula for fourth derivatives of the function $I_{\vec X}(\vec \lambda)$ in Theorem \ref{main}.
\end{example}

\begin{example}
From Examples \ref{ex:der2}, \ref{ex:der3} and \ref{ex:der4}, we recover the following results of \cite{guo2011estimation} for the scalar Gaussian channel (\ref{sgc}):
\begin{align*}
I_X^{(2)}(\lambda) &= \frac{1}{2} \E \sq{-M_2^2}\\
I_X^{(3)}(\lambda) &= \frac{1}{2} \E \sq{ 2 M_2^3 - M_3^2 } \\
I_X^{(4)}(\lambda) &= \frac{1}{2} \E \sq{ -15 M_2^4 + 12 M_3^2 M_2 + 6 M_4 M_2^2 -  M_4^2 },
\end{align*}
where
\begin{align*}
    M_k = \E[ (X - \E[X|Y])^k | Y ].
\end{align*}
\end{example}

\begin{remark}
    To list all diverse partitions in the expansion of $ \tau$ or $\bar \tau$, it is useful to take into account the one-to-one correspondence between diverse partitions and loop-free graphs (graphs with no vertex connecting to itself). Let $ \pi$ be a diverse partition of the multiset $[n]_2$. Each $ i \in [n]$ must belong to two different blocks of $\pi$, and we connect these two blocks by an edge labeled by $ i$. For example, the the partition with blocks $\br{1,2}$, $ \br{3,4}$, $\br{1,2,3,4}$ corresponds to the following graph:

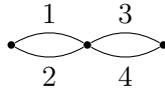
\begin{center}
    \begin{tikzpicture}
        \coordinate[circle, inner sep=1pt, fill=black](1) at (0,0);
        \coordinate[circle, inner sep=1pt, fill=black](2) at (1,0);
        \coordinate[circle, inner sep=1pt, fill=black](3) at (2,0);
        \draw (1) to[bend left=30] node[above]{1} (2);
        \draw (1) to[bend right=30] node[below]{2} (2);
        \draw (2) to[bend left=30] node[above]{3}(3);
        \draw (2) to[bend right=30] node[below]{4} (3);
\end{tikzpicture}
\end{center}
Conversely, given a loop-free graph with edges labeled by $[n]$, we can recover the corresponding diverse partition by looking at the edges that connect to each vertex.

The lines in the expansion in Example \ref{ex:der4} correspond to the graphs in Figure \ref{fig:gr1}. Moreover, the terms in each line correspond to different ways of labeling the edges of the associated graph. For example, the six terms in the third line of the expansion corresponds to different labelings in Figure \ref{fig:gr2}. Note that two labelings are considered the same if they describe the same partition (Figure \ref{fig:gr3}).
\end{remark}

\begin{figure}[h]
\centering
  \begin{minipage}[t]{0.18\linewidth}
    \centering
    \begin{tikzpicture}[baseline]
        \coordinate[circle, inner sep=1pt, fill=black](1) at (0,-0.1);
        \coordinate[circle, inner sep=1pt, fill=black](2) at (1,-0.1);
        \coordinate[circle, inner sep=1pt, fill=black](3) at (1,0.9);
        \coordinate[circle, inner sep=1pt, fill=black](4) at (0,0.9);
        \draw (1) to (2);
        \draw (2) to (3);
        \draw (3) to (4);
        \draw (4) to (1);
    \end{tikzpicture}
  \end{minipage}\hfill
  \begin{minipage}[t]{0.18\linewidth}
    \centering
    \begin{tikzpicture}[baseline]
        \coordinate[circle, inner sep=1pt, fill=black](1) at (0,0);
        \coordinate[circle, inner sep=1pt, fill=black](2) at (1,0);
        \coordinate[circle, inner sep=1pt, fill=black](3) at (1,0.8);
        \coordinate[circle, inner sep=1pt, fill=black](4) at (0,0.8);
        \draw (1) to[bend left] (2);
        \draw (1) to[bend right] (2);
        \draw (3) to[bend left] (4);
        \draw (3) to[bend right] (4);
    \end{tikzpicture}
   \end{minipage}\hfill
   \begin{minipage}[t]{0.18\linewidth}
    \centering
    \begin{tikzpicture}[baseline]
        \coordinate[circle, inner sep=1pt, fill=black](1) at (0,0);
        \coordinate[circle, inner sep=1pt, fill=black](2) at (1,0);
        \coordinate[circle, inner sep=1pt, fill=black](3) at (0.5,0.866);
        \draw (1) to[bend left] (2);
        \draw (1) to[bend right] (2);
        \draw (2) to (3);
        \draw (3) to (1);
    \end{tikzpicture}
  \end{minipage}
  \begin{minipage}[t]{0.18\linewidth}
    \centering
    \begin{tikzpicture}[baseline]
        \coordinate[circle, inner sep=1pt, fill=black](1) at (0,-0.1);
        \coordinate[circle, inner sep=1pt, fill=black](2) at (1,-0.1);
        \coordinate[circle, inner sep=1pt, fill=black](3) at (0.5, 0.9);
        \draw (3) to[bend left=20] (1);
        \draw (3) to[bend right=20] (1);
        \draw (3) to[bend left=20] (2);
        \draw (3) to[bend right=20] (2);
    \end{tikzpicture}
  \end{minipage}
    \begin{minipage}[t]{0.18\linewidth}
    \centering
    \begin{tikzpicture}[baseline]
        \coordinate[circle, inner sep=1pt, fill=black](1) at (0,-0.1);
        \coordinate[circle, inner sep=1pt, fill=black](2) at (0,0.9);
        \draw (1) to[bend left=20] (2);
        \draw (1) to[bend right=20] (2);
        \draw (1) to[bend left=50] (2);
        \draw (1) to[bend right=50] (2);
    \end{tikzpicture}
  \end{minipage}
\caption{Loop-free graphs with four edges.}
\label{fig:gr1}
\end{figure}
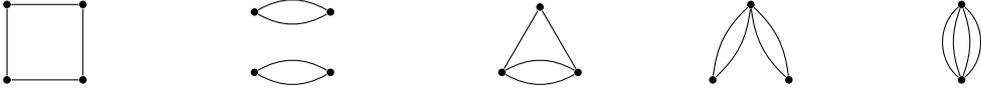

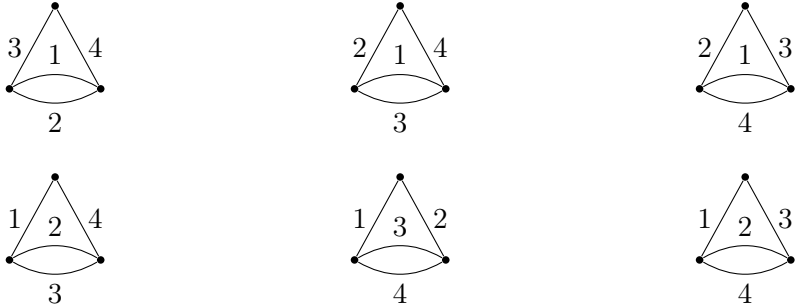
\begin{figure}[h]
	\begin{minipage}[t]{0.3\linewidth}
    \centering
    \begin{tikzpicture}[baseline]
        \coordinate[circle, inner sep=1pt, fill=black](1) at (0,0);
        \coordinate[circle, inner sep=1pt, fill=black](2) at (1.2,0);
        \coordinate[circle, inner sep=1pt, fill=black](3) at (0.6,1.1);
        \draw (1) to[bend left] node[above]{1} (2);
        \draw (1) to[bend right] node[below]{2} (2);
        \draw (1) to node[left]{3} (3) ;
        \draw (2) to node[right]{4} (3);
    \end{tikzpicture}
	\end{minipage}
	\begin{minipage}[t]{0.3\linewidth}
    \centering
    \begin{tikzpicture}[baseline]
        \coordinate[circle, inner sep=1pt, fill=black](1) at (0,0);
        \coordinate[circle, inner sep=1pt, fill=black](2) at (1.2,0);
        \coordinate[circle, inner sep=1pt, fill=black](3) at (0.6,1.1);
        \draw (1) to[bend left] node[above]{1} (2);
        \draw (1) to[bend right] node[below]{3} (2);
        \draw (1) to node[left]{2} (3) ;
        \draw (2) to node[right]{4} (3);
    \end{tikzpicture}
	\end{minipage}
	\begin{minipage}[t]{0.3\linewidth}
    \centering
    \begin{tikzpicture}[baseline]
        \coordinate[circle, inner sep=1pt, fill=black](1) at (0,0);
        \coordinate[circle, inner sep=1pt, fill=black](2) at (1.2,0);
        \coordinate[circle, inner sep=1pt, fill=black](3) at (0.6,1.1);
        \draw (1) to[bend left] node[above]{1} (2);
        \draw (1) to[bend right] node[below]{4} (2);
        \draw (1) to node[left]{2} (3) ;
        \draw (2) to node[right]{3} (3);
    \end{tikzpicture}
	\end{minipage}

	\vspace{1em}

	\begin{minipage}[t]{0.3\linewidth}
    \centering
    \begin{tikzpicture}[baseline]
        \coordinate[circle, inner sep=1pt, fill=black](1) at (0,0);
        \coordinate[circle, inner sep=1pt, fill=black](2) at (1.2,0);
        \coordinate[circle, inner sep=1pt, fill=black](3) at (0.6,1.1);
        \draw (1) to[bend left] node[above]{2} (2);
        \draw (1) to[bend right] node[below]{3} (2);
        \draw (1) to node[left]{1} (3) ;
        \draw (2) to node[right]{4} (3);
    \end{tikzpicture}
	\end{minipage}
	\begin{minipage}[t]{0.3\linewidth}
    \centering
    \begin{tikzpicture}[baseline]
        \coordinate[circle, inner sep=1pt, fill=black](1) at (0,0);
        \coordinate[circle, inner sep=1pt, fill=black](2) at (1.2,0);
        \coordinate[circle, inner sep=1pt, fill=black](3) at (0.6,1.1);
        \draw (1) to[bend left] node[above]{3} (2);
        \draw (1) to[bend right] node[below]{4} (2);
        \draw (1) to node[left]{1} (3) ;
        \draw (2) to node[right]{2} (3);
    \end{tikzpicture}
	\end{minipage}
	\begin{minipage}[t]{0.3\linewidth}
    \centering
    \begin{tikzpicture}[baseline]
        \coordinate[circle, inner sep=1pt, fill=black](1) at (0,0);
        \coordinate[circle, inner sep=1pt, fill=black](2) at (1.2,0);
        \coordinate[circle, inner sep=1pt, fill=black](3) at (0.6,1.1);
        \draw (1) to[bend left] node[above]{2} (2);
        \draw (1) to[bend right] node[below]{4} (2);
        \draw (1) to node[left]{1} (3) ;
        \draw (2) to node[right]{3} (3);
    \end{tikzpicture}
	\end{minipage}
\caption{Different labelings of a graph.}
\label{fig:gr2}
\end{figure}

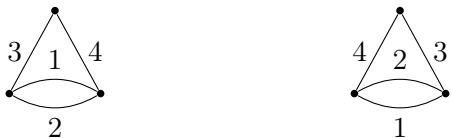
\begin{figure}[h]
\centering
	\begin{minipage}[t]{0.3\linewidth}
    \centering
    \begin{tikzpicture}[baseline]
        \coordinate[circle, inner sep=1pt, fill=black](1) at (0,0);
        \coordinate[circle, inner sep=1pt, fill=black](2) at (1.2,0);
        \coordinate[circle, inner sep=1pt, fill=black](3) at (0.6,1.1);
        \draw (1) to[bend left] node[above]{1} (2);
        \draw (1) to[bend right] node[below]{2} (2);
        \draw (1) to node[left]{3} (3) ;
        \draw (2) to node[right]{4} (3);
    \end{tikzpicture}
	\end{minipage}
	\begin{minipage}[t]{0.3\linewidth}
    \centering
    \begin{tikzpicture}[baseline]
        \coordinate[circle, inner sep=1pt, fill=black](1) at (0,0);
        \coordinate[circle, inner sep=1pt, fill=black](2) at (1.2,0);
        \coordinate[circle, inner sep=1pt, fill=black](3) at (0.6,1.1);
        \draw (1) to[bend left] node[above]{2} (2);
        \draw (1) to[bend right] node[below]{1} (2);
        \draw (1) to node[left]{4} (3) ;
        \draw (2) to node[right]{3} (3);
    \end{tikzpicture}
	\end{minipage}
\caption{Example of two equivalent labelings, describing the same partition $(\br{3,4}, \br{1,2,3}, \br{1,2,4})$ }
\label{fig:gr3}
\end{figure}

\section{Tools}\label{tools}
We present here the main tools of the paper: the equivalent of Gaussian channels and the replica method. The equivalent result implies that we can compute any derivative of the mutual information if we know how to compute expressions of the form $\partial_{\lambda_1} \dots \partial_{\lambda_n} I_{\vec X}(\vec 0)$. The replica method, on the other hand, is used to obtain a combinatorial formula for  $\partial_{\lambda_1} \dots \partial_{\lambda_n} I_{\vec X}(\vec 0)$.
\subsection{A nice property of Gaussian channels}
\begin{prop}\label{reduce}
    A set of Gaussian channels with the same signal $X$ and independent noises is equivalent to a single Gaussian channel with signal $X$ and SNR equal to the sum of individual SNRs.
\end{prop}
\begin{proof}
The proof relies on the concept of sufficient statistics and its connections to information theory \cite{cover1999elements}. Suppose the channels are
\begin{align}\label{osd}
    Y_i = \sqrt{\lambda_i} X+Z_i, \quad i=1 \ddd n,
\end{align}
where $Z_i$ are independent standard Gaussian noises independent of $X$. The posterior distribution of $ X$ given $(Y_1 \ddd Y_n)$ is
\begin{align}
P_{X|\vec Y}(dx) \propto P_X(dx)\exp \bround{ \sum_i \sqrt{\lambda_i} Y_i x - \frac{1}{2} \lambda_i x^2 },
\end{align}
which implies that
\begin{align}
     S := \sum_{i=1}^n \sqrt{\lambda_i}  Y_i
\end{align}
is a sufficient statistics for estimating $X$ from the outputs. Moreover, let $\lambda = \sum_i \lambda_i$, then $S/\sqrt{\lambda}$ can be written as $\sqrt{\lambda}  X + \xi$, where
\begin{align}
    \xi = \sum_{i=t}^n \sqrt{\lambda_i/\lambda} \, Z_i,
\end{align}
is independent of $X$ and follows the standard normal distribution. The lemma follows from the fact that $S/\sqrt{\lambda}$, which contains all the information relevant to $X$ that can be inferred from the outputs, is the output of a Gaussian channel with SNR $ \lambda$.
\end{proof}

\begin{remark}
Proposition \ref{reduce} implies that
\begin{align}
&I(X; \sqrt{\lambda_1}X + Z_1 \ddd \sqrt{\lambda_n}X + Z_n) \nn = &I(X; \sqrt{\lambda_1+\dots+\lambda_n} X + Z),
\end{align}
where $Z,Z_1 \ddd Z_n$ are i.i.d as $\mathcal N(0,1)$ and independent of the random variable $X$. This formula is at the base of the incremental channel approach in \cite{guo2005mutual} and \cite{guo2011estimation}. Our proof based on sufficient statistics gives a cleaner derivation of this result compared to the proofs given in these references, which rely on the language of signal processing.
\end{remark}

\subsection{Replica method}\label{replica}
\subsubsection{A quick introduction to the replica method}\label{cu}
The replica method, which plays a central role in deriving our result, is based on the following simple formula:
\begin{align}
\log Z = \partial_{r=0} Z^r.
\end{align}
In practice, the replica computations are performed assuming $r$ to be an integer, but the final result is obtained by sending $r$ to zero. Despite this lack of rigor, the replica method can derive results in a much more straightforward manner compared to other methods. Originating from statistical physics, the replica method has been widely used to analyze large information processing systems, notably in \cite{guo2005randomly}, \cite{moustakas2003mimo} and \cite{tulino2013support}. In our paper, however, the replica method is applied to a finite system and does not involves either high dimensional limits or the choice of replica ansatz. Instead, the computations are largely combinatorial. Moreover, our main result is not derived entirely from the replica computations, but also from the incremental channel approach.

As a quick illustration of the method, we present here a short and probably new derivation of the classic cumulant-moment relation (\ref{mc}) using the replica method.

By the replica trick, we have
\begin{align}\label{obo}
\psi_{\vec X}(\vec \lambda) =& \partial_{r=0} \bsq{ \E \exp \bround{ \sum_{i=1}^n \lambda_i X_i } }^r.
\end{align}
By treating $r$ as if it were an integer, we can write
\begin{align}\label{obo1}
\psi_{\vec X}(\vec \lambda) =& \partial_{r=0} \E \exp \bround{ \sum_{i=1}^n \lambda_i \sum_{a=1}^r X_{ia} },
\end{align}
where  $(X_{ia})_{i=1}^n$ for $a \in [r]$ are independent random vectors with the same distribution as $(X_1 \ddd X_n)$. Applying $\partial_{\lambda_1} \dots \partial_{\lambda_n}$ at $\vec \lambda = \vec 0$ on both sides of (\ref{obo1}), we obtain
\begin{align}
\kappa(X_1 \ddd X_n) &= \partial_{r=0}\E \prod_{i=1}^n \sum_{a=1}^r X_{ia},
\end{align}
where on the right hand side, we exchanged the operator $\partial_{\lambda_1} \dots \partial_{\lambda_n}$ with $\partial_{r=0} \E$. Expanding the product on the right hand side of this equation and bringing the expectation inside the sum, we have
\begin{align}\label{hjj}
\kappa(X_1 \ddd X_n) &= \partial_{r=0} \sum_{a_1 \ddd a_n=1}^r \E[X_{1a_1} X_{2a_2} \dots X_{na_n}].
\end{align}
Next we will construct the following map from $[r]^n$ to the set of all partitions of $[n]$. Given $(a_1 \ddd a_n) \in [r]^n$, a unique partition $\pi$ of $[n]$ can be obtained by putting $i$ and $j$ in the same block if $a_i = a_j$. Suppose that the blocks of $\pi$ are $B_1 \ddd B_k$, we have  
\begin{align}\label{dfd}
	\E[X_{1a_1} \dots X_{n a_n}] &= \E[X_{B_1}] \dots \E[X_{B_k}],
\end{align}
which follows from the fact that $X_{ia_i}$ and $X_{ja_j}$ are independent if $a_i \neq a_j$.
On the other hand, for each partition $\pi$ of $[n]$ with $k$ blocks, there are
\begin{align}
    P_{r,k} := r(r-1)\dots(r-k+1)
\end{align}
elements in $[r]^n$ that are mapped to $\pi$. Therefore, equation (\ref{hjj}) can be rewritten as
\begin{align}\label{qpl}
    \kappa(X_1 \ddd X_n) = \partial_{r=0} \sum_{\pi} P_{r,k} \, \E[X_{B_1}] \dots \E[X_{B_k}],
\end{align}
where the sum runs over all partitions $\pi = (B_1 \ddd B_k)$ of $ [n]$ (thus $k$ can take any value from $1$ to $n$ as $\pi$ runs over all partitions of $[n]$). The cumulant-moment relation (\ref{mc}) follows from the fact that
\begin{align}
	\partial_{r=0} P_{r,k}  = (-1)^{k-1}(k-1)!.
\end{align}

\begin{remark}
Given $(a_1 \ddd a_n) \in [r]^n$, the number of blocks $k$ in (\ref{dfd}) cannot be larger than $r$. However, in (\ref{qpl}), the sum is over all the partitions of $[n]$. There is no contradiction here because for a partition with $k$ blocks with $k>r$, the coefficient $P_{r,k}$ in (\ref{qpl}) is zero.
\end{remark}

\subsubsection{Mutual information and replicas} Let $X$, $Y$ be random variables with values in $\mathcal X$ and $\mathcal Y$ respectively. Suppose that $\mathcal X$ and $\mathcal Y$ are equipped with measure $\mu$ and $\nu$, called the \emph{underlying measure}. Let $p_X, p_Y, p_{X,Y}$ be the density functions of the random variables $X$, $Y$, $(X,Y)$ with respect the underlying measures $\mu, \nu, \mu \otimes \nu$. Denote $p(y|x) = p_{Y|X}(y|x)$ for simplicity. By definition, the mutual information between $ X$ and $ Y$ is
\begin{align}
	I(X;Y) &= \E \log p(Y|X) - \E \log p_Y(Y).
\end{align}
Note that the mutual information does not depend on the choice of underlying measures. 

Next, we have
\begin{align}
    \E \log p_Y(Y) &= \int \nu(dy) p_Y(y) \log p_Y(y) \\
    &= \partial_{r=1} \int \nu(dy) \, p_Y(y)^r.
\end{align}
Treating $r$ as if it were an integer, we write
\begin{align}
    p_Y(y)^r = \E[p(y|X)]^r = \E[p(y|X_1) \dots p(y|X_r)],
\end{align}
where $X_a$ for $a\in[r]$ are independent and identically  distributed as $X$. As a result, we obtain the following replica representation of $I(X;Y)$
\begin{align}\label{eq:main}
I(X;Y) &= \E \log  p(Y|X) \nn
&- \partial_{r=1} \E \int \nu(dy) p(y|X_1)\dots p(y|X_r).
\end{align}

\section{Proofs}\label{pf}

\begin{lemma}\label{gg}
With the setting of Theorem \ref{main}, for $n \geq 2$,
\begin{align}
\partial_{\lambda_1} \dots \partial_{\lambda_n} I_{\vec X}(\vec \lambda)|_{\vec \lambda = \vec 0} = \tau(X_1 \ddd X_n).
\end{align}
\end{lemma}

\begin{proof}
We will apply formula (\ref{eq:main}) in which $(X, Y)$ is replaced by $(\vec X, \vec Y)$ in this lemma. Let $\vec X_a = (X_{ia})_{i=1}^n$ for $a \in [r]$ be independent random vectors with the same distribution as $\vec X = (X_1 \ddd X_n)$. In the context of (\ref{eq:main}), we have $\mathcal X = \mathcal Y = \R^n$. By choosing $\mu$, the underlying measure of $ \mathcal X$, to be the Lebesgue measure and $\nu$, the underlying measure of $\mathcal Y$, to be the standard Gaussian measure, we have
\begin{align}
    p(\vec y|\vec x) = \exp \bround{ \sum_{i=1}^n \sqrt{\lambda_i} y_i x_i - \frac{1}{2}\lambda_i x_i^2 }.
\end{align}
Therefore,
\begin{align}
\E \log p(\vec Y|\vec X) &= \E[ \sum_{i=1}^n \sqrt{\lambda_i} Y_i X_i - \frac{1}{2}\lambda_i X_i^2 ]\nn
&= \frac{1}{2} \sum_{i=1}^n \lambda_i \E[X_i^2],
\end{align}
and
\begin{align}
&\E \int \nu(d\vec y) p(\vec y|\vec X_1)\dots p(\vec y|\vec X_r) \nn
=& \E \int \nu(d\vec y) \exp \bround{ \sum_{a=1}^r  \sum_{i=1}^n \sqrt{\lambda_i} y_i X_{ia} - \frac{1}{2}\lambda_i X_{ia}^2 } \nn
=& \E \exp \bround{ \sum_{i=1}^n \lambda_i \sum_{1\leq a < b \leq r}  X_{ia} X_{ib} },
\end{align}
where the last equality is obtained by performing the Gaussian integral over $\vec y$. We thus obtain from (\ref{eq:main}) the following replica representation of $I_{\vec X}(\vec \lambda)$
\begin{align} \label{op}
I_{\vec X}(\vec \lambda) &= \frac{1}{2} \sum_{i=1}^n \lambda_i \E[X_i^2] \nn  &- \partial_{r=1} \E \exp \bround{ \sum_{i=1}^n \lambda_i \sum_{1\leq a < b \leq r}  X_{ia} X_{ib} },
\end{align}

The rest of the proof closely follows the derivation of the cumulant-moment formula in Section \ref{cu}. By applying $\partial_{\lambda_1} \dots \partial_{\lambda_n}$ at $\vec \lambda = \vec 0$ on both sides of (\ref{op}), we obtain
\begin{align}
    &\partial_{\lambda_1} \dots \partial_{\lambda_n} I_{\vec X}(\vec \lambda)|_{\vec \lambda = \vec 0} \nn
    = &-\partial_{r=1} \E \bsq{ \prod_{i=1}^n \sum_{1 \leq a <b\leq r} X_{ia} X_{ib}  }
\end{align}
where on the right hand side, we exchanged the operator $\partial_{\lambda_1} \dots \partial_{\lambda_n}$ with $\partial_{r=0} \E$. Expanding the product on the right hand side of this equation and bringing the expectation inside the sum, we have
\begin{align}\label{ukk}
    &\partial_{\lambda_1} \dots \partial_{\lambda_n} I_{\vec X}(\vec \lambda)|_{\vec \lambda = \vec 0} \nn
    = &-\partial_{r=1} \sum_{\substack{1 \leq a_1 < b_1 \leq r, \\ \dots, \\ 1 \leq a_n < b_n \leq r}} \E [ X_{1a_1}X_{1b_1} \dots X_{na_n} X_{nb_n}].
\end{align}
We can rewrite this sum into a sum over diverse partitions of $[n]_2$ as follows. Let us call $a_1, b_1, \ddd a_n, b_n$ the \emph{replica indices}. Since $X_{ia_i}$ and $X_{ja_j}$ are independent if $a_i \neq a_j$, we have 
\begin{align}\label{dfb}
	\E[X_{1a_1} X_{1b_1} \dots X_{na_n} X_{nb_n}] &= \E[X_{B_1}] \dots \E[X_{B_k}],
\end{align}
where $\pi = (B_1 \ddd B_k)$ is the partition of the multiset $[n]_2$ obtained by putting $i,j \in [n]_2$ in the same block if and only if the corresponding replica indices are equal. Since $a_i < b_i$ for all $ i \in [n]$, the partition $\pi$ is diverse. On the other hand, this mapping from replica indices to partitions is many-to-one: each diverse partition $\pi$ of $[n]_2$ with $k$ blocks corresponds to
\begin{align}
2^{-s(\pi)}P_{r,k}
\end{align}
choices of replica indices. Here, the factor $ 2^{-s(\pi)}$ accounts for the fact that some of the blocks of $ \pi$ are identical. As a result, equation (\ref{ukk}) can be rewritten as
\begin{align}
    \partial_{\lambda_1} \dots \partial_{\lambda_n} I_{\vec X}(\vec \lambda)|_{\vec \lambda = \vec 0} = -\partial_{r=1} \sum_{\pi} \frac{P_{r,k}}{2^{s(\pi)}} \, \E[X_{B_1}] \dots \E[X_{B_k}],
\end{align}
where the sum is over all diverse partitions $ \pi = (B_1 \ddd B_k)$ of the multiset $[n]_2$.  The result of the lemma follows from
\begin{align}
    \partial_{r=1} P_{r,k} = (-1)^k (k-2)!.
\end{align}
\end{proof}

\begin{lemma}\label{oi} With the setting of Theorem \ref{main}, for non-negative integers $k_1 \ddd k_n$ whose sum is greater than $1$, we have
\begin{align}\label{hf1}
   \partial_{\lambda_1}^{k_1} \dots \partial_{\lambda_n}^{k_n} I_{\vec X}(\vec \lambda)|_{\vec \lambda = \vec 0}  = \tau(\underbrace{X_1 \ddd X_1}_{k_1}, \dots, \underbrace{X_n \ddd X_n}_{k_n} ).
\end{align}
\end{lemma}

\begin{proof}
First we consider the case where $ k_1 \ddd k_n \geq 1$. Consider the following Gaussian channels with independent noises, with signals 
\begin{align}
    \underbrace{X_1 \ddd X_1}_{k_1} \ddd \underbrace{X_n \ddd X_n}_{k_n},
\end{align}
and SNRs given in order as
\begin{align}
    \lambda_{11} \ddd \lambda_{1 k_1} \ddd \lambda_{n1} \ddd \lambda_{n k_n}.
\end{align}
By Proposition \ref{reduce}, these channels can be reduced without information loss into the following channels with signals
\begin{align}
    X_1 \ddd X_k,
\end{align}
and SNRs respectively as
\begin{align}
    \lambda_{11} + \dots + \lambda_{1 k_1} \ddd \lambda_{n1} + \dots + \lambda_{n k_n}.
\end{align}
From this equivalence, we have
\begin{align}
&I_{X_1 \ddd X_n}(\lambda_{11} + \dots + \lambda_{1 k_1} \ddd \lambda_{n1} + \dots + \lambda_{n k_n}) \nn
=&I_{X_1 \ddd X_1 \ddd X_n \ddd X_n}(\lambda_{11} \ddd \lambda_{1 k_1} \ddd \lambda_{n1} \ddd\lambda_{n k_n}) .
\end{align}
The claim of the lemma follows from the previous equation by taking the first derivative of each variable at zero and using Lemma \ref{gg}.

For the case in which there exists $i \in [n]$ such that $ {k_i=0}$, without loss of generality, suppose that 
$ k_1 \ddd k_m \geq 1$ and $ k_{m+1} \ddd k_n =0$ for some $ m < n$. It is clear that $I_{\vec X}(\lambda_1 \ddd \lambda_m, 0 \ddd 0) = I_{X_1 \ddd X_m}(\lambda_1 \ddd \lambda_m)$. Therefore,
\begin{align}
    &\partial_{\lambda_1}^{k_1} \dots \partial_{\lambda_n}^{k_n} I_{\vec X}(\vec \lambda)|_{\vec \lambda = \vec 0}\nn
	= &\partial_{\lambda_1}^{k_1} \dots \partial_{\lambda_m}^{k_m} I_{\vec X}(\lambda_1 \ddd \lambda_m, 0 \ddd ,0)|_{\lambda_1=\dots=\lambda_m=0} \nn
    = &\partial_{\lambda_1}^{k_1} \dots \partial_{\lambda_m}^{k_m} I_{X_1 \ddd X_m}(\lambda_1 \ddd \lambda_m)|_{\lambda_1 = \dots = \lambda_m =0} \nn
    = &\tau(\underbrace{X_1 \ddd X_1}_{k_1}, \dots, \underbrace{X_m \ddd X_m}_{k_m} ),
\end{align}
which proves the lemma in this case.

\end{proof}

\noindent \textbf{Proof of Proposition \ref{lt}.} 
a) In the definition of $\tau$ given by (\ref{mc}), for each partition $\pi$, the function 
\begin{align}
(X_1 \ddd X_n) \ra \E[X_{B_1}] \dots \E[X_{B_{|\pi|}}]
\end{align}
is multiquadratic, since each $X_i$ appears twice in this expression.  Because a sum of multiquadratic forms (with the same number of arguments) is also multiquadratic, $\tau$ is multiquadratic. Similarly, $\bar \tau$ is multiquadratic.

\vspace{.3cm}

b) If $[n]$ can be divided into two non-empty sets $I$ and $J$ such that $ (X_i)_{i \in I}$ and $ (X_j)_{j \in J}$ are independent, then
\begin{align}
    I_{\vec X}(\vec \lambda) = I_{(X_i)_{i \in I}}((\lambda_i)_{i \in I}) + I_{(X_j)_{j \in J}}((\lambda_j)_{j \in J}),
\end{align}
since the noises are independent. By taking the derivative $ \partial_{\lambda_1} \dots \partial_{\lambda_n}$ at $ \vec \lambda=0$ on both sides of the previous equation and using Lemma \ref{gg}, we obtain
\begin{align}
    \tau(X_1 \ddd X_n)=0.
\end{align} 

c) Let $\bar X_i = X_i - \E[X_i]$. Since mutual information is invariant by translation, we have
\begin{align}
    I_{\vec X}(\vec \lambda) = I_{\vec X - \vec c}(\vec \lambda),
\end{align}
for any $ \vec c \in \R^n$. By this and by Lemma \ref{gg}, we have
\begin{align}
\tau(X_1 \ddd X_n) &= \partial_{\lambda_1} \dots \partial_{\lambda_n} I_{\vec X}(\vec \lambda)|_{\vec \lambda = \vec 0} \quad \nn
&= \partial_{\lambda_1} \dots \partial_{\lambda_n} I_{\vec X - \E[\vec X]}(\vec \lambda)|_{\vec \lambda = \vec 0} \nn
&= \tau(\bar X_1 \ddd \bar X_n) \label{dg3}.
\end{align} 
In the expansion of $\tau(\bar X_1 \ddd \bar X_n)$, for any partition $\pi = (B_1 \ddd B_{|\pi|})$ that contains a block of size one, we have $\E[\bar X_{B_1}] \dots \E[\bar X_{B_{|\pi|}}] = 0$. Therefore, the expansion only involves partitions with all block sizes larger than one. Consequently,
\begin{align}\label{dg4}
\tau(\bar X_1 \ddd \bar X_n) &= \bar \tau(\bar X_1 \ddd \bar X_n).
\end{align}
The result follows from (\ref{dg3}) and (\ref{dg4}).
\qed

\vspace{.5cm}

\noindent \textbf{Proof of Theorem \ref{main}.} We will only prove equation (\ref{jj}) of the theorem, as equation (\ref{jj1}) can be proved similarly. Our proof will follow the same incremental channel approach of \cite{guo2011estimation}. Suppose the following data is given in addition to $\vec Y$,
\begin{align*}
    \vec Y' = \sqrt{\vec \delta} \odot \vec X + \vec Z',
\end{align*}
where the noise $\vec Z'$ is standard Gaussian, independent of $\vec X$ and $\vec Z$. Since the two channels that share the signal $X_i$ with SNRs $\lambda_i $ and $ \delta_i$ can be combined into a single channel with SNR $\lambda_i+\delta_i$, we have
\begin{align}
    I(\vec X; \vec Y, \vec Y') = I_{\vec X}(\vec \lambda + \vec \delta).
\end{align}
Thus,
\begin{align}
    I_{\vec X}(\vec \lambda+\vec \delta) - I_{\vec X}(\vec \lambda) &= I(\vec X;\vec Y, \vec Y') - I(\vec X; \vec Y) \nn
    &= I(\vec X;\vec Y'|\vec Y) \nn
    &= \int P_{\vec Y}(d \vec y) I(\vec X;\vec Y'|\vec Y = \vec y).
\end{align}
Now taking the derivative $\partial_{\delta_1}^{k_1} \dots \partial_{\delta_n}^{k_n}$ at $\vec \delta = 0$ on both sides of this equation and exchange the derivative with the integral, we obtain
\begin{align}\label{lk1}
    &\partial_{\lambda_1}^{k_1} \dots \partial_{\lambda_n}^{k_n} I_{\vec X}(\vec \lambda) \nn
    = &\int P_{\vec Y}(d \vec y) \, \partial_{\delta_1}^{k_1} \dots \partial_{\delta_n}^{k_n} I(\vec X;\vec Y'|\vec Y = \vec y)|_{\vec \delta =\vec 0}.
\end{align}
Let $\vec X^{\vec y} $ be the random variable $\vec X$ conditioned on the event $ \vec Y = \vec y$. Since $\vec Z'$ is independent of $ \vec Y$, we have
\begin{align}
    I(\vec X;\vec Y'|\vec Y = \vec y) &= I(\vec X; \sqrt{\vec \delta}\odot \vec X + \vec Z' | \vec Y = \vec y ) \nn
    &= I(\vec X^{\vec y} ; \sqrt{\vec \delta} \odot \vec X^{\vec y} + \vec Z' ).
\end{align}
From this and (\ref{lk1}), we have
\begin{align}\label{nb}
    &\partial_{\lambda_1}^{k_1} \dots \partial_{\lambda_n}^{k_n} I_{\vec X}(\vec \lambda) \nn
    = &\int P_{\vec Y}(d \vec y) \, \partial_{\delta_1}^{k_1} \dots \partial_{\delta_n}^{k_n}  I(\vec X^{\vec y} ; \sqrt{\vec \delta}\odot \vec X^{\vec y} + \vec Z' )|_{\vec \delta = \vec 0}.
\end{align}
Since $\vec Z'$ is independent of $\vec X$ and $\vec Y$, it is also independent of $ \vec X^{\vec y}$. By Lemma \ref{oi}, we have
\begin{align}
    &\partial_{\delta_1}^{k_1} \dots \partial_{\delta_n}^{k_n}I(\vec X^{\vec y} ; \sqrt{\vec \delta}\odot \vec X^{\vec y} + \vec Z' )|_{\vec \delta = \vec 0} \nn
    =& \tau( \underbrace{X_1^{\vec y} \ddd X_1^{\vec y} }_{k_1} \ddd \underbrace{X_n^{\vec y} \ddd X_n^{\vec y} }_{k_n} ) \nn
    =& \tau( \underbrace{X_{1} \ddd X_{1} }_{k_1} \ddd \underbrace{X_{n} \ddd X_{n} }_{k_n} | \vec Y = \vec y ).
\end{align}
From this and (\ref{nb}) we obtain the equation (\ref{jj}) in the theorem. \qed

\section{Conclusion}
We derived a general formula for the derivatives of the mutual information with respect to the SNRs in vector Gaussian channels, by combining the incremental channel approach with the replica method. The result can be written compactly by a form $\tau$ defined as a combinatorial sum. The form  $\tau$ and the joint cumulants exhibit some remarkable similarities, summarized in the following table:

\begin{table*}[!htb]
\begin{center}
\normalsize
\begin{tabular}{ c | c }\label{tab}
$\kappa$ & $\tau$ \\
 \hline
multilinear & multiquadratic\\
$\kappa(X_1\ddd X_n)= \partial_{\lambda_1}\dots \partial_{\lambda_n}\psi_{\vec X}(\vec \lambda)|_{\vec \lambda=0}$ & $\tau(X_1\ddd X_n)=\partial_{\lambda_1}\dots \partial_{\lambda_n}I_{\vec X}(\vec \lambda)|_{\vec \lambda=0}$ \\
sum over partitions of $[n]$ & sum over diverse partitions of $[n]_2$\\
\hline
\multicolumn{2}{c}{do not depend on the order of arguments}\\
\multicolumn{2}{c}{vanish if the arguments can be divided into two independent parts}\\
\hline
\end{tabular}
\end{center}
\end{table*}

\section*{Acknowledgment}
We are grateful to Steeve Zozor and Walid Hachem for their careful reading and feedback on the manuscript, and Olivier Michel for pointing out the relevant papers on Independence Component Analysis.

\bibliographystyle{siam}
\bibliography{main.bib}

\end{document}